\documentclass[runningheads]{llncs}
\usepackage[english]{babel}
\usepackage[utf8]{inputenc}
\usepackage{amsmath,amssymb,dsfont}
\usepackage{mathtools}
\usepackage{thmtools}
\usepackage{graphicx}
\usepackage{float}
\usepackage{tikz}
\usetikzlibrary{calc}
\usepackage{tkz-graph}
\usetikzlibrary{arrows,calc,decorations.pathreplacing,shapes,decorations.pathmorphing}

\usepackage[final]{pdfpages}

\title{Hardness of some variants of the graph coloring game}

\author{
   Thiago Marcilon \inst{1}\and
   Nicolas Martins \inst{2}\and
   Rudini Sampaio\inst{3}
}

\institute{
   Centro de Ciências e Tecnologia, Univ. Federal do Cariri, Juazeiro do Norte, Brazil.\\\email{thiago.marcilon@ufca.edu.br}
   \and
   Univ. Integração Internacional Lusofonia Afrobrasileira Unilab, Redenção, Brazil.\\\email{nicolasam@unilab.edu.br}
   \and
   Departamento de Computa\c c\~ao, Universidade Federal do Cear\'a, Fortaleza, Brazil.\\\email{rudini@dc.ufc.br}
}

\begin{document}
\mainmatter
\maketitle

\begin{abstract}
Very recently, a long-standing open question proposed by Bodlaender in 1991 was answered: the graph coloring game is PSPACE-complete. In 2019, Andres and Lock proposed five variants of the graph coloring game and left open the question of PSPACE-hardness related to them. In this paper, we prove that these variants are PSPACE-complete for the graph coloring game and also for the greedy coloring game, even if the number of colors is the chromatic number. Finally, we also prove that a connected version of the graph coloring game, proposed by Charpentier et al. in 2019, is PSPACE-complete.
\end{abstract}


\section{Introduction}\label{intro}

In the graph coloring game, given a graph $G$ and a set $C$ of integers (representing the color set), two players (Alice and Bob) alternate turns (starting with Alice) in choosing an uncolored vertex to be colored by an integer of $C$ not already assigned to one of its colored neighbors. In the greedy coloring game, there is one additional constraint: the vertices must be colored by the least possible integer of $C$. Alice wins if all vertices are successfully colored. Otherwise, Bob wins the game.
From the classical Zermelo-von Neumann theorem, one of the two players has a winning strategy, since it is a finite game without draw.
Thus, the game chromatic number $\chi_g(G)$ and the game Grundy number $\Gamma_g(G)$ are defined as the least numbers of colors in the set $C$ for which Alice has a winning strategy in the graph coloring game and the greedy coloring game, respectively.

Clearly, $\chi(G)\leq\chi_g(G)$ and $\chi(G)\leq\Gamma_g(G)\leq\Gamma(G)$, where $\chi(G)$ is the chromatic number of $G$ and $\Gamma(G)$ is the Grundy number of $G$ (the maximum number of colors that can be used by a greedy coloring of $G$).

The graph coloring game was first considered by Brams about 38 years ago in the context of coloring maps and was described by Gardner in 1981 in his ``\emph{Mathematical Games}'' column of Scientific American \cite{gardner81}. It remained unnoticed until Bodlaender \cite{bodlaender91} reinvented it in 1991.

Since then, the graph coloring game became a very active topic of research.
In 1993, Faigle et al. \cite{faigle93} proved that $\chi_g(G)\leq 4$ in forests and, in 2007, Sidorowicz \cite{game-cactus07} proved that $\chi_g(G)\leq 5$ in cacti. In 1994, Kierstead and Trotter \cite{kierstad94} proved that $\chi_g(G)\leq 7$ in outerplanar graphs. In 1999, Dinski and Zhu proved that $\chi_g(G)\leq k(k+1)$ for every graph with acyclic chromatic number $k$ \cite{dinski-zhu99}. In 2000, Zhu proved that $\chi_g(G)\leq 3k+2$ in partial $k$-trees \cite{zhu00}. For planar graphs, Zhu \cite{zhu08} proved in 2008 that $\chi_g(G)\leq 17$, Sekiguchi \cite{sekiguchi14} proved in 2014 that $\chi_g(G)\leq 13$ if the girth is at least 4 and Nakprasit et al. \cite{nakprasit18} proved in 2018 that $\chi_g(G)\leq 5$ if the girth is at least 7. In 2008, Bohman, Frieze and Sudakov \cite{bohman08} investigated the asymptotic behavior of $\chi_g(G_{n,p})$ for the random graph $G_{n,p}$.

In Bodlaender's 1991 paper, the complexity was left as ``\textit{an interesting open problem}''. A point of difficulty to set the complexity is that it is not clear what its natural decision problem is. As pointed out by Zhu \cite{zhu99}, the graph coloring game ``\emph{exhibits some strange properties}'' and the following naive question is still open (Question 1 of \cite{zhu99}): Does Alice have a winning strategy for the coloring game with $k+1$ colors if she has a winning strategy with $k$ colors?
Thus it is possible to define two decision problems for the graph coloring game: given a graph $G$ and an integer $k$: $\chi_g(G)\leq k$\ ? Does Alice have a winning strategy with $k$ colors? Both problems are equivalent if and only if Question 1 of \cite{zhu99} is true.

Nevertheless, it was proved in 2019 that both coloring game decision problems are PSPACE-complete \cite{rudini19}, solving Bodlaender's 1991 question. Also in 2019, Andres and Lock \cite{andres-lock19} proposed five variants of the graph coloring game: $g_B$ (Bob starts the game), $g_{A,A}$ (Alice starts and can pass turns), $g_{A,B}$ (Alice starts and Bob can pass turns), $g_{B,A}$ (Bob starts and Alice can pass turns) and $g_{B,B}$ (Bob starts and can pass turns). They left the following problem: ``\textit{the question of PSPACE-hardness remains open for all the game variants mentioned above}''.

In 2019, Charpentier, Hocquard, Sopena and Zhu \cite{zhu19} proposed a connected version of the graph coloring game (starting with Alice): the subgraph induced by the set of colored vertices must be connected. They prove that Alice wins with 2 colors in bipartite graphs and with 5 colors in outerplanar graphs.

In 2013, Havet and Zhu \cite{havet13} proposed the greedy coloring game and the game Grundy number $\Gamma_g(G)$. They proved that $\Gamma_g(G)\leq 3$ in forests and $\Gamma_g(G)\leq 7$ in partial 2-trees. They also posed two questions. Problem 5 of \cite{havet13}: $\chi_g(G)$ can be bounded by a function of $\Gamma_g(G)$? Problem 6 of \cite{havet13}: Is it true that $\Gamma_g(G)\leq\chi_g(G)$ for every graph $G$?
In 2015, Krawczyk and Walczak \cite{poset15} answered Problem 5 of \cite{havet13} in the negative: $\chi_g(G)$ is not upper bounded by a function of $\Gamma_g(G)$. To the best of our knowledge, Problem 6 of \cite{havet13} is still open.
In 2019, it was proved that the greedy coloring game is PSPACE-complete \cite{rudini19}.
It was also proved that the game Grundy number is equal to the chromatic number in split graphs and extended $P_4$-laden graphs, even if Bob starts and can pass any turn.

In this paper, we prove that all variants of the graph coloring game and the greedy coloring game are PSPACE-complete even if the number of colors is the chromatic number for any pair $Y\in\{Alice, Bob\}$ and $Z\in\{Alice, Bob, No\ one\}$, where $Y$ starts the game and $Z$ can pass turns, by reductions from POS-CNF-11 and POS-DNF-11. Finally, we also prove that the connected version of the graph coloring game is PSPACE-complete, by a reduction from the variant of the graph coloring game in which Bob starts the game.

\section{PSPACE-complete variants of graph coloring game}

Firstly, let us consider the game variant $g_B$: Bob starts the game.
Let $\chi_g^B(G)$ be the minimum number of colors in the set $C$ for which Alice has a winning strategy in $g_B$.
Clearly, $\chi_g^B(G)\geq\chi(G)$. With this, we can define two decision problems for $g_B$: given a graph $G$ and an integer $k$,
\begin{itemize}
\item (Problem $g_B$-1) $\chi_g^B(G)\leq k$\ ?
\item (Problem $g_B$-2) Does Alice have a winning strategy in $g_B$ with $k$ colors?
\end{itemize}

In this section, we prove that the following more restricted problem is PSPACE-complete: given a graph $G$ and its chromatic number $\chi(G)$,
\begin{itemize}
\item (Problem $g_B$-3) $\chi_g^B(G)=\chi(G)$\ ?
\end{itemize}
Notice that $\chi(G)$ is part of the input of Problem $g_B$-3.

It is easy to see that Problems $g_B$-1 and $g_B$-2 are generalizations of Problem $g_B$-3, since both problems are equivalent to it for $k=\chi(G)$. For this, notice that $\chi_g^B(G)\leq k=\chi(G)$ if and only if $\chi_g^B(G)=\chi(G)$, which is true if and only if Alice has a winning strategy in $g_B$ with $k=\chi(G)$ colors. Then the PSPACE-hardness of Problem $g_B$-3 implies the PSPACE-hardness of Problems $g_B$-1 and $g_B$-2.
To the best of our knowledge, no paper have explicitly defined these decision problems or proved pertinence in PSPACE. 

\begin{lemma}\label{lem.pspace}
Problems $g_B$-1, $g_B$-2 and $g_B$-3 are in PSPACE.
\end{lemma}

\begin{proof}[Sketch]
Let $G$ be a graph with $n$ vertices and $k\leq n$ be an integer. 
Let us begin with Problem $g_B$-2. Since the number of turns is exactly $n$ and, in each turn, the number of possible moves is at most $n\cdot k$ (there are at most $n$ vertices to select and at most $k$ colors to use), we have that Problem $g_B$-2 is a polynomially bounded two player game and then it is in PSPACE \cite{demaine09}.
Consequently, Problem $g_B$-3 is also in PSPACE.

Finally, regarding Problem $g_B$-1, notice that it can be decided using Problem $g_B$-2 for all $k'=\chi(G),\ldots,k$. That is, if there is $k'\in\{\chi(G),\ldots,k\}$ such that Problem $g_B$-2 with $k'$ colors is YES, then Problem $g_B$-1 is also YES. Otherwise, it is NO. Since Problem $g_B$-2 is in PSPACE, then Problem $g_B$-1 is also in PSPACE.
\end{proof}

Now, we prove that Problem $g_B$-3 is PSPACE-complete.
In \cite{rudini19}, the PSPACE-hardness reduction of the graph coloring game used the POS-CNF problem, which is known to be log-complete in PSPACE \cite{poscnf78}. In POS-CNF, we are given a set $\{X_1,\ldots,X_N\}$ of $N$ variables and a CNF formula (conjunctive normal form: conjunction of disjunctions) with $M$ clauses $C_1,\ldots,C_M$ (also called disjunctions), in which only positive variables appear (that is, no negations of variables). Players I and II alternate turns setting a previously unset variable True or False, starting with Player I. After all $N$ variables are set, Player I wins if and only if the formula is True. Clearly, since there are only positive variables, we can assume that Players I and II always set variables True and False, respectively.

Unfortunately, by associating Player I with Alice and Player II with Bob, all our attempts to obtain a reduction for $g_B$ similar to the one in \cite{rudini19} using POS-CNF have failed. 
However, another problem proved to be usefull for $g_B$: POS-DNF, which is also PSPACE-complete \cite{poscnf78}. In POS-DNF, we are given a DNF formula (disjunctive normal form: disjuntion of conjunctions) instead of a CNF formula.
In other words, Player I in POS-DNF has a similar role of Player II in POS-CNF: he wins if plays every variable of some conjunction. From now on, we will call Players I and II of POS-DNF as Bob and Alice, respectively.
As an example, consider the DNF formula $(X_1\wedge X_2\wedge X_5)\vee(X_1\wedge X_3\wedge X_5)\vee(X_2\wedge X_4\wedge X_5)\vee(X_3\wedge X_4\wedge X_5)$. Note that Bob has a winning strategy for this formula firstly setting $X_5$ True, since it is in all conjunctions: if Alice sets $X_1$ False, Bob sets $X_4$ True; if Alice sets $X_4$ False, Bob sets $X_1$ True; if Alice sets $X_2$ False, Bob sets $X_3$ True; if Alice sets $X_3$ False, Bob sets $X_2$ True.

\begin{lemma}\label{prop1}
If Bob (resp. Alice) has a winning strategy in POS-CNF or POS-DNF, then he (resp. she) also has a winning strategy if Alice (resp. Bob) can pass any turn.
\end{lemma}

\begin{proof}[Sketch]
In both cases, if the opponent passed a turn, just assume that the opponent has selected some non-selected variable and keep playing following the winning strategy in the original game. If the opponent selects this assumed variable later in the game, just assume that other non-selected variable was selected and keep playing with the winning strategy. If all variables (including the assumed ones) were selected, then the formula is true and just select any assumed variable (in any order). With this, since the player have followed a winning strategy in POS-CNF or POS-DNF and the assumptions restricted only the player (and not the opponent), we are done.
\end{proof}

If the disjunctions/conjunctions have at most 11 variables, we are in POS-CNF-11 and POS-DNF-11 problems, which are also PSPACE-complete \cite{poscnf78}.

One important ingredient of the reduction is the graph $F_1$ of Figure \ref{fig:F1}, which has a clique $K$, an independent set $Q$ of $|K|+3$ vertices and three vertices $s$, $w$ and $y$ such that $s$ and $w$ are neighbors and are adjacent to all vertices in $K$ and $y$ is adjacent to all vertices in $K\cup Q$. We start proving that, in case of Bob firstly coloring $s$, Alice must color $y$ with the same color of $s$ in her first move.

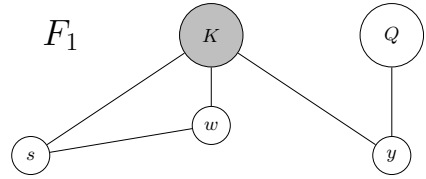
\begin{figure}[h]\centering\scalebox{0.8}{
	\begin{tikzpicture}
	\tikzstyle{vertex}=[draw,circle,fill=white!25,minimum size=18pt,inner sep=1pt]
	\tikzstyle{svertex}=[draw,circle,fill=white!25,minimum size=10pt,inner sep=1pt]
	\tikzstyle{clique}=[draw,circle,fill=black!25,minimum size=30pt,inner sep=1pt]
	\tikzstyle{stable}=[draw,circle,fill=black!00,minimum size=30pt,inner sep=1pt]

	\node(s) at (-2.5,2) {\LARGE{$F_1$}};
	\node[vertex](s) at (-3,0) {$s$};
	\node[vertex](w) at (0,0.5) {$w$};
	\node[vertex](y) at (3,0) {$y$};
	\node[clique](K) at (0,2) {$K$};
	\node[stable](Q) at (3,2) {$Q$};
	
	\draw (s) -- (w) -- (K);
	\draw (s) -- (K);
	\draw (y) -- (K);
	\draw (y) -- (Q);
	
	\end{tikzpicture}}
	\caption{\label{fig:F1}Graph $F_1$: clique $K$ with $k$ vertices and independent set $Q$ with $k+3$ vertices.}
\end{figure}

\begin{lemma}\label{lemF1}
Consider the graph $F_1$ of Figure \ref{fig:F1} with $|K|=k$ and $|Q|=k+3$ and assume that Bob colored vertex $s$ in the first move of $g_B$. Alice wins the game in $F_1$ with $k+2$ colors if and only if she colors vertex $y$ with the same color of $s$ in her first move.
\end{lemma}

\begin{proof}[sketch]
Without loss of generality, assume $s$ was colored with color 1.
Note that $F_1$ can be colored with $k+2$ colors if and only if either $y$ and $s$ or $y$ and $w$ receive the same color. Then, she wins with $k+2$ colors if colors $y$ with 1 in her first move. Thus assume Alice does not color $y$ with 1 in her first move.

During the game, we say that a vertex $v$ sees a color $c$ if $v$ has a neighbor colored $c$. Bob can win by coloring $s, w$ and $y$ with distinct colors in the following way. He firstly colors a vertex of $Q$ with color 1, avoiding Alice to color $y$ with color 1. Now, Bob has to guarantee that $y$ and $w$ receive different colors. For this, the following strategy holds:
(i) if $w$ is not colored and some color $c\ne 1$ appears in $Q$ and does not appear in $K$, then he colors $w$ with color $c$;
(ii) If $w$ is colored, $y$ is not colored and does not see some color $c$ distinct from the color of $w$, then he colors $y$ with $c$;
(iii) otherwise, he colors any vertex of $F_1$ preferring vertices of $Q$ with any color not appearing in the neighborhood of $y$.

This strategy guarantees that every color seen by $w$ is also seen by $y$. Moreover, after a Bob's move from (iii), he guarantees that some color $c$ seen by $y$ is not seen by $w$. Thus Alice must color a vertex of $K$ with $c$, since otherwise Bob wins in his turn from (i) or (ii). Since $|Q|=|K|+3$, Alice cannot do this indefinitely and Bob wins the game.
\end{proof}

\begin{theorem}\label{teo1}
Given a graph $G$, deciding whether $\chi_g^B(G) = \chi(G)$ is PSPACE-complete. Thus, given $k$, deciding whether $\chi_g^B(G) \leq k$ or deciding if Alice has a winning strategy in $g_B$ with $k$ colors are PSPACE-complete problems.
\end{theorem}

\begin{proof}[Sketch]
From Lemma \ref{lem.pspace}, the three decision problems are in PSPACE.
Given a POS-DNF-11 formula with $N$ variables $X_1,\ldots,X_N$ and $M$ conjunctions $C_1,\ldots,C_M$, let $p_j$ (for $j=1,\ldots,M$) be the size of conjunction $C_j$ ($p_j\leq 11$). We will construct a graph $G$ such that $\chi(G) = M+3N+25$ and $\chi_g(G) = M+3N+25$ if and only if Alice has a winning strategy for the POS-DNF-11 formula.

Initially, the constructed graph $G$ is the graph $F_1$ of Figure 1 with $|K|=M+3N+23$ and $|Q|=|K|+3$. See Figure \ref{fig1b}.
For every variable $X_i$, create a vertex $x_i$ in $G$. For every conjunction $C_j$, we create a \emph{conjunction clique}. For this, first create a clique with vertices $\ell_{j,1},\ldots,\ell_{j,p_j}$ and join $\ell_{j,k}$ to $x_i$ with an edge if and only if both are associated to the same variable, for $k=1,\ldots,p_j$. Also add the new vertex $\ell_{j,0}$ (which is not associated to variables) and join it with an edge to the vertex $y$. For every vertex $\ell_{j,k}$ ($j=1,\ldots,M$ and $k=0,\ldots,p_j$), replace it by two true-twin vertices $\ell'_{j,k}$ and $\ell''_{j,k}$, which are adjacent vertices with same neighborhood of $\ell_{j,k}$. Moreover, add to the conjunction clique of $C_j$ a clique $L_j$ with size $M+3N+25-2(p_j+1)\geq 3N$ and join all vertices of $L_j$ to $s$. With this, all conjunction cliques have exactly $M+3N+25$ vertices.

Figure \ref{fig1b} shows the constructed graph $G$ for the formula $(X_1\wedge X_2\wedge X_5)\vee(X_1\wedge X_3\wedge X_5)\vee(X_2\wedge X_4\wedge X_5)\vee(X_3\wedge X_4\wedge X_5)$. Recall that Bob has a winning strategy in POS-DNF-11 firstly setting $X_5$ True: if Alice sets $X_1$ False, Bob sets $X_4$ True; if Alice sets $X_4$ False, Bob sets $X_1$ True; if Alice sets $X_2$ False, Bob sets $X_3$ True; if Alice sets $X_3$ False, Bob sets $X_2$ True. In the reduction of this example, we have $N=5$ variables, $M=4$ conjunctions, $p_j=3$, $|K|=42$, $|Q|=45$, the cliques $L_1$ to $L_4$ have $M+3N+25-2(p_j+1)=36$ vertices each.

\begin{figure}[h]\label{fig1b}\centering\scalebox{0.8}{
\begin{tikzpicture}[auto]
\tikzstyle{vertex}=[draw,circle,fill=white!25,minimum size=8pt,inner sep=1pt]
\tikzstyle{vertex1}=[draw,circle,fill=white!25,minimum size=12pt,inner sep=2pt]
\tikzstyle{vertex2}=[draw,circle,fill=black!15,minimum size=12pt,inner sep=3pt]
\tikzstyle{vertex3}=[draw,circle,fill=black!25,minimum size=20pt,inner sep=7pt]
\tikzstyle{vertex4}=[draw,circle,fill=black!00,minimum size=20pt,inner sep=7pt]
\node[vertex2] (s) at (2.5,5)   {$s$};
\node[vertex]  (w) at (4.5,5.5) {$w$};
\node[vertex3] (K) at (4.5,7)   {$K$};
\node[vertex4] (Q) at (6.5,7)   {$Q$};
\node[vertex2] (y) at (6.5,5)   {$y$};
\node () at (2.7,6.5) {\LARGE{$F_1$}};

\node[vertex] (x1) at (0,-1) {$x_1$};
\node[vertex] (x2) at (3,-1) {$x_2$};
\node[vertex] (x3) at (6,-1) {$x_3$};
\node[vertex] (x4) at (9,-1) {$x_4$};
\node[vertex] (x5) at (4.5,-2.2) {$x_5$};

\node[vertex] (L11) at (-0.8,1){$\ell_{1,1}$};
\node[vertex] (L12) at (0.8,1) {$\ell_{1,2}$};
\node[vertex] (L15) at (0,0.3) {$\ell_{1,3}$};
\node[vertex] (L21) at (2.2,1) {$\ell_{2,1}$};
\node[vertex] (L22) at (3.8,1) {$\ell_{2,2}$};
\node[vertex] (L25) at (3,0.3) {$\ell_{2,3}$};
\node[vertex] (L31) at (5.2,1) {$\ell_{3,1}$};
\node[vertex] (L32) at (6.8,1) {$\ell_{3,2}$};
\node[vertex] (L35) at (6,0.3) {$\ell_{3,3}$};
\node[vertex] (L41) at (8.2,1) {$\ell_{4,1}$};
\node[vertex] (L42) at (9.8,1) {$\ell_{4,2}$};
\node[vertex] (L45) at (9,0.3) {$\ell_{4,3}$};

\node[vertex2] (L13) at (-0.5,2){$L_{1}$};
\node[vertex] (L14) at (0.5,2) {$\ell_{1,0}$};
\node[vertex2] (L23) at (2.5,2) {$L_{2}$};
\node[vertex] (L24) at (3.5,2) {$\ell_{2,0}$};
\node[vertex2] (L33) at (5.5,2) {$L_{3}$};
\node[vertex] (L34) at (6.5,2) {$\ell_{3,0}$};
\node[vertex2] (L43) at (8.5,2) {$L_{4}$};
\node[vertex] (L44) at (9.5,2) {$\ell_{4,0}$};

\path[-]
(K) edge (s) edge (y) edge (w) (w) edge (s)
(x1) edge[bend left] (L11) edge (L21)
(x2) edge (L12) edge (L31)
(x3) edge (L22) edge (L41)
(x4) edge (L32) edge[bend right] (L42)
(x5) edge[bend left] (L15) edge (L25) edge (L35) edge[bend right] (L45)
(L11)edge(L12)edge(L13)edge(L14)edge(L15)(L12)edge(L13)edge(L14)edge(L15)(L13)edge(L14)edge(L15)(L14)edge(L15)
(L21)edge(L22)edge(L23)edge(L24)edge(L25)(L22)edge(L23)edge(L24)edge(L25)(L23)edge(L24)edge(L25)(L24)edge(L25)
(L31)edge(L32)edge(L33)edge(L34)edge(L35)(L32)edge(L33)edge(L34)edge(L35)(L33)edge(L34)edge(L35)(L34)edge(L35)
(L41)edge(L42)edge(L43)edge(L44)edge(L45)(L42)edge(L43)edge(L44)edge(L45)(L43)edge(L44)edge(L45)(L44)edge(L45)
(s) edge(L13) edge(L23) edge(L33) edge(L43)
(y) edge(L14) edge(L24) edge(L34) edge(L44) edge(Q);
\end{tikzpicture}}
\caption{Constructed graph $G$ for the formula $(X_1\wedge X_2\wedge X_5)\vee(X_1\wedge X_3\wedge X_5)\vee(X_2\wedge X_4\wedge X_5)\vee(X_3\wedge X_4\wedge X_5)$. Recall that each vertex $\ell_{j,k}$ represents two true-twins $\ell'_{j,k}$ and $\ell''_{j,k}$; $L_1$, $L_2$, $L_3$, $L_4$ are cliques with $36$ vertices; $K$ is a clique with $42$ vertices. Bob has a winning strategy avoiding $44$ colors in the graph coloring game.}
\end{figure}
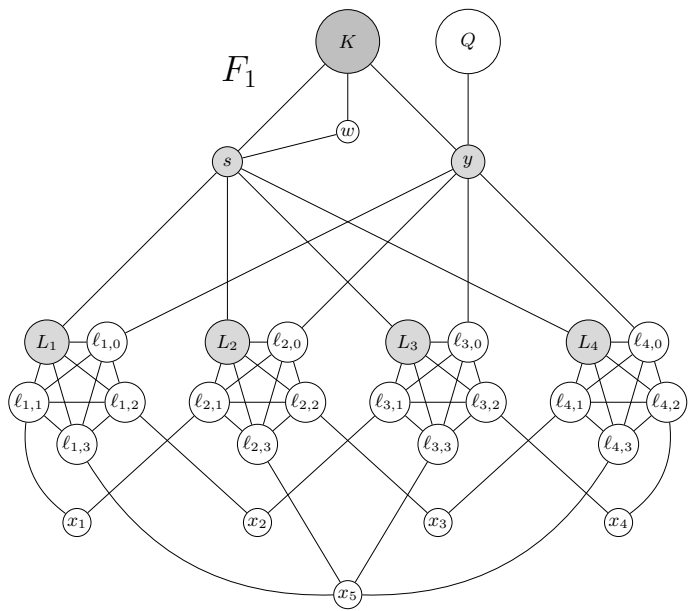

It is easy to check that $\chi(G) = M+3N+25$. For this, color $s$ and all vertices in $Q$ with color 1, the vertices of $K$ with colors $2$ to $M+3N+24$, color $w$, $y$ and every vertex $x_i$ ($i=1,\ldots,n$) with color $M+3N+25$. For every $j=1,\ldots,M$, color the vertices $\ell'_{j,k}$ and $\ell''_{j,k}$ with colors $2k+1$ and $2k+2$ ($k=0,\ldots,p_j$).	Finally, color the vertices of the clique $L_j$ using the colors $2p_j+3,\ldots,M+3N+25$. Since the conjunction cliques contains $M+3N+25$ vertices, then $\chi(G) = M+3N+25$.

In the following, we show that Alice has a winning strategy in the graph coloring game if and only if she has a winning strategy in POS-DNF-11.
From Lemma \ref{lemF1}, in her first move, Alice must color vertex $y$ of $F_1$ if Bob colored vertex $s$ in his first move. 
Roughly speaking, we show that, in the best strategies, Bob colors vertex $s$ first and Alice colors vertex $y$ with the same color. Also notice that every vertex of a conjunction clique has degree exactly $M+3N+25$ (since it has exactly one neighbor outside the clique).
In order to have the conjunction clique colored using the colors $1,\ldots,M+3N+25$, Alice must guarantee that all colors appearing in the outside neighbors of a conjunction clique also appears inside the clique. On the other hand, we show that Bob's strategy is making all outside neighbors of a conjunction clique to be colored with the same color of $s$ and $y$ (which will represent True in POS-DNF-11) and thus impeding Alice of using this color inside the conjunction clique.

We first show that if Bob has a winning strategy in POS-DNF-11, then $\chi_g(G) > M+3N+25$.
Assume that Bob wins in POS-DNF-11. Bob uses the following strategy. He firstly colors $s$ with color $1$ and, from Lemma \ref{lemF1}, Alice must color $y$ with $1$. In the next rounds, Bob follows his first POS-DNF-11 winning strategy: colors with color 1 the vertex associated to the variable that should receive True. If Alice colors a vertex in $N[x_i]$ (the closed neighborhood of $x_i$) for some $i$, Bob considers that she marked $X_i$ False in POS-DNF-11 and he follows his winning POS-DNF-11 strategy; if Alice does not color any vertex in $N[x_i]$ for some $i$, then Bob plays as if Alice has passed her turn in POS-DNF-11 (recall Lemma \ref{prop1}). Then at some point all literals of some conjunction will be marked True. This means that all outside neighbors of some conjunction clique will be colored with color $1$. Since the clique has $M+3N+25$ vertices and color $1$ cannot be used, we have that $\chi_g(G)>M+3N+25$.

We now show that if Alice has a winning strategy in the POS-DNF-11 game then $\chi_g(G)=M+3N+25$. Assume that Alice wins in the POS-DNF-11 game.

Firstly suppose Bob colors $s$ (resp. $y$), say with color 1, in his first move.
Then Alice must color $y$ (resp. $s$) with color 1 in her first move (recall Lemma \ref{lemF1}).
Alice can play using the following strategy: (1) if Bob plays on $x_i$, Alice plays as if Bob has chosen $X_i$ to be True in POS-DNF-11, meaning that she colors the vertex $x_j$ with a color different from $1$, where $X_j$ is the literal chosen by her winning strategy in POS-DNF-11; (2) if Bob plays on some twin obtained from vertex $\ell_{i,j}$, Alice plays the least available color in the other twin; (3) otherwise, Alice plays as if Bob has passed his turn in POS-DNF-11 (recall Lemma \ref{prop1}) if this game is not over yet; otherwise colors any non-colored vertex of $G$ with the least available color.
Following this strategy, every conjunction clique has a vertex colored $1$. Since each clique $L_j$ has at least $3N$ vertices, Alice and Bob can finish coloring every conjunction clique using the colors $1,\ldots,M+3N+25$.

Now assume Bob colored $v_1\not\in\{s,y\}$ in his first move (with some color $c$).
Then Alice colors $y$ firstly with a color $c'\ne c$, say $c'=2$ w.l.g.
Let $v_2$ be the 2nd vertex chosen by Bob. If $v_2=s$ and its color is 2, we are done from the last paragraph (just replacing color 2 by color 1). Otherwise, Alice colors $w$ with color 2 (and then $s$ cannot be colored 2). We show that Alice has a winning strategy in this case. Assume w.l.g. that the color of $s$ in the game will be 1 (otherwise we can relabel the colors). Thus no vertex of $L_j$ is colored 1 ($j=1,\ldots,M$).

With this, if Bob colored $\ell'_{i,0}$ or $\ell''_{i,0}$ for some $i$ and the corresponding conjunction clique does not have a vertex colored 1, then Alice must color a vertex inside this conjunction clique with color 1.
Otherwise, if there is a non-colored variable vertex, Alice colors it with a color distinct from 1.
Since each clique $L_j$ has at least $3N$ vertices, then Bob cannot color all vertices of some $L_j$ before all variable vertices are colored.
With this, Alice can guarantee that all colors of the variable vertices appear in the conjunction cliques and Alice wins.
\end{proof}

Following the same path of $g_B$, we define $\chi_g^{A,A}(G)$, $\chi_g^{A,B}(G)$, $\chi_g^{B,A}(G)$ and $\chi_g^{B,B}(G)$: the minimum number of colors in $C$ s.t Alice has a winning strategy in $g_{A,A}$, $g_{A,B}$, $g_{B,A}$ and $g_{B,B}$, resp.
We can also define three decision problems for each game $g_{Y,Z}$ ($Y,Z\in\{A,B\}$): given a graph $G$, its chromatic number $\chi(G)$ and an integer $k$,
\begin{itemize}
\item (Problem $g_{Y,Z}$-1) $\chi_g^{Y,Z}(G)\leq k$\ ?
\item (Problem $g_{Y,Z}$-2) Does Alice have a winning strategy in $g_{Y,Z}$ with $k$
colors?
\item (Problem $g_{Y,Z}$-3) $\chi_g^{Y,Z}(G)=\chi(G)$\ ?
\end{itemize}

\begin{corollary}\label{corol1}
For every $Y,Z\in\{A,B\}$, the decision problems $g_{Y,Z}$-1, $g_{Y,Z}$-2 and $g_{Y,Z}$-3 are PSPACE-complete.
\end{corollary}

\begin{proof}[Sketch]
Let $Y,Z\in\{A,B\}$. Following similar arguments in Lemma \ref{lem.pspace}, we obtain that they are PSPACE. The crucial argument to prove PSPCE-hardness is Lemma \ref{prop1}, which asserts that if Bob (resp. Alice) has a winning strategy in POS-CNF or POS-DNF-11, then he (resp. she) also has a winning strategy if Alice (resp. Bob) can pass any turn. Following the proof of Theorem 2.2 in \cite{rudini19} if $Y=A$ or the proof of Theorem \ref{teo1} above if $Y=B$, we have that a winning strategy in $g_{Y}$ is obtained from a winning strategy in POS-CNF / POS-DNF-11. If Alice (resp. Bob) has a winning strategy in the related logical game and $Z=B$ (resp. $Z=A$), she (resp. he) also has a winning strategy in $g_{Y,Z}$ by following the winning strategy in the logical game when the opponent passes a turn. Now if Alice (resp. Bob) has a winning strategy in the related logical game and $Z=A$ (resp. $Z=B$), she (resp. he) also has a winning strategy in $g_{Y,Z}$ (just not passing moves and simulating the obtained winning strategy in $g_Y$).

\end{proof}

\section{Connected graph coloring game is PSPACE-complete}

In this section, we prove that the connected version of the graph coloring game \cite{zhu19} is PSPACE-complete with a reduction from Problem $g_B$-3 of Section 2.

\begin{theorem}
Given a graph $G$ and an integer $k$, deciding whether Alice has an winning strategy with exactly $k$ colors or at most $k$ colors in the connected version of the graph coloring game (Alice starting) are PSPACE-complete problems.
\end{theorem}

\begin{proof}[sketch]
As before, we first define a more restricted decision problem and prove that it is PSPACE-complete: given $G$ and its chromatic number $\chi(G)$, Alice has a winning strategy with $\chi(G)$ colors?

We obtain a reduction from Problem $g_{B}$-3 of Section 2. 
Let $(G,\chi(G))$ be an instance of Problem $g_{B}$-3 with $|V(G)|$ odd.
The reduction is as depicted in Figure \ref{fig1z}, where $K$ is a clique with size $\chi(G)$ and $s$ is connected to every vertex of $G$. Notice that $|V(G)\cup\{s\}|$ is even.

\begin{figure}[h]\label{fig1z}\centering\scalebox{0.75}{
	\centering
	\begin{tikzpicture}
	\tikzstyle{vertex}=[draw,circle,fill=white!25,minimum size=18pt,inner sep=1pt]
	\tikzstyle{svertex}=[draw,circle,fill=white!25,minimum size=10pt,inner sep=1pt]
	\tikzstyle{clique}=[draw,circle,fill=black!25,minimum size=30pt,inner sep=1pt]
	
	\node(F) at (-4,-1) {\LARGE{$F_2$}};
	\node[vertex](y1) at (-4,0) {$y_1$};
	\node[vertex](y2) at (0,0) {$y_2$};
	\node[vertex](s) at (-2,-2) {$s$};
	\node[clique](K) at (-2,0) {$K$};
	\node(G) at (1,-2) {\LARGE{$G$}};
	
	\draw (s) -- (y1) -- (K) -- (y2) -- (s) -- (G);
	
	\node(F) at (3,-1) {\LARGE{$F_2$}};
	\node[vertex](ly1) at (3,0) {$y_1$};
	\node[vertex](ly2) at (7,0) {$y_2$};
	\node[vertex](p) at (9,0) {$p$};
	\node[vertex](ls) at (5,-2) {$s$};
	\node[clique](lK) at (5,0) {$K$};
	\node(lG) at (8,-2) {\LARGE{$G$}};

	\draw (ls) -- (ly1) -- (lK) -- (ly2) -- (ls) -- (lG);
	\draw (ly2) -- (p);
	\end{tikzpicture}}
	\caption{\label{fig1z}The reduction from the graph $G$ adding a gadget $F_2$ to it. The left if $\chi(G)$ is even or the right if $\chi(G)$ is odd.}
\end{figure}
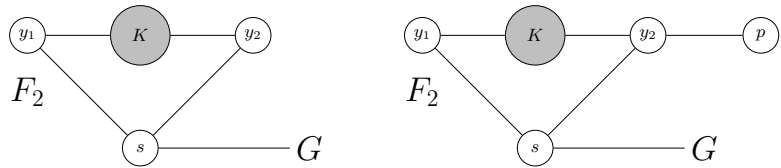

The resulting graph $G'$ has chromatic number $\chi(G')=\chi(G) + 1$.
Note that, if $y_1$ and $y_2$ receive distinct colors, Bob wins since it is impossible to color $G'$ with $\chi(G')$ colors.

Also, note that, if Alice does not play her first move in $y_1, y_2$ or $K$, she loses. This is because, since $|V(G)\cup\{s\}|$ is even, Bob can always guarantee that Alice will play on either $y_1$ or $y_2$ before him and consequently he can play a different color in the other, forcing distinct colors for $y_1$ and $y_2$. This is true even if she plays in $p$ first or in $y_2$ first when $p$ is a vertex of $G'$.

First, assume that Alice has a winning strategy for the variant $g_B$-3 of the graph coloring game (Bob starts the game). She has the following strategy: she begins playing $y_1$. Bob has to play in either $s$ or $K$. She then plays in $y_2$ with the same color as $y_1$, ensuring her safety inside $F_2$. From there on, if Bob plays in $G$, she plays according to her strategy in $g_B$-3. If he plays in $F_2$, she also plays in $F_2$ (since $|V(F)|$ is odd, she can always do this).

Now, assume that Bob has a winning strategy in $g_B$-3. Assume, without loss of generality, that Alice plays either $y_1$ or $K$. If she plays in $y_1$, Bob can play in $s$ and then she has to play in $y_2$. If she plays in $K$, he can play in $y_1$ and then she has to play in $y_2$. From now on, in either case, Bob can guarantee he is the first to make a move in $G$ since $|V(F)|$ is odd. After this, if she plays in $G$, he also plays in $G$ following his winning strategy. If she plays in $F_2$, he also plays in $F_2$ which is always possible.
\end{proof}

\section{PSPACE-complete variants of greedy coloring game}

As in Section 2, we define five variants of the greedy coloring game: $g^*_B$ (Bob starts), $g^*_{A,A}$ (Alice starts and can pass any turn), $g^*_{A,B}$ (Alice starts and Bob can pass any turn), $g^*_{B,A}$ (Bob starts and Alice can pass any turn) and $g^*_{B,B}$ (Bob starts and can pass any turn).
Unlike in the game coloring problem, the greedy game coloring problem satisfies the following:

\begin{proposition}
If Alice (resp. Bob) has a winning strategy with $k$ colors in $g^*_{Y,Z}$, then she (resp. he) also has a winning strategy with $k+1$ colors ($Y,Z\in\{A,B\}$).
\end{proposition}

\begin{proof} A winning strategy with $k$ colors in the greedy coloring game is a strategy with $k+1$ colors that does not use the color $k+1$, since the coloring is greedy.
\end{proof}

Let us start with $g^*_B$ (Bob starts the greedy coloring game).
Let $\Gamma_g^B(G)$ be the minimum number of colors in $C$ s.t Alice has a winning strategy in $g^*_B$. Clearly, $\chi(G)\leq\Gamma_g^B(G)\leq\Gamma(G)$.
We can define two natural decision problem for $g^*_B$: given a graph $G$, its chromatic number $\chi(G)$ and an integer $k$,
\begin{itemize}
\item (Problem $g^*_B$-1) $\Gamma^B_g(G)\leq k$? Alice has winning strategy with $k$ colors in $g^*_B$?
\item (Problem $g^*_B$-2) $\Gamma^B_g(G)=\chi(G)$?
\end{itemize}

Clearly, Problem $g^*_B$-1 is a generalization of Problem $g^*_B$-2 (just set $k=\chi(G)$).
Then the PSPACE-hardness of $g^*_B$-2 implies the PSPACE-hardness of $g^*_B$-1.



We obtain a reduction from POS-DNF-11 similar to the one of Section 2 for $g_B$.
One important ingredient of the reduction is the graph $F_3$ of Figure \ref{fig:F3}, which has a clique $K$ with $k$ vertices and three vertices $s$, $w$ and $y$ such that $s$ and $w$ are adjacent to all vertices in $K$ and $w$ is adjacent to $y$. We start proving that, in case of Bob firstly coloring $s$, Alice must color $y$ in her first move.

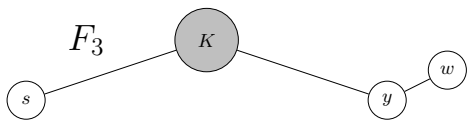
\begin{figure}[h]\centering\scalebox{0.8}{
	\begin{tikzpicture}
	\tikzstyle{vertex}=[draw,circle,fill=white!25,minimum size=18pt,inner sep=1pt]
	\tikzstyle{svertex}=[draw,circle,fill=white!25,minimum size=10pt,inner sep=1pt]
	\tikzstyle{clique}=[draw,circle,fill=black!25,minimum size=30pt,inner sep=1pt]
	\tikzstyle{stable}=[draw,circle,fill=black!00,minimum size=30pt,inner sep=1pt]

	\node(s) at (-2,1) {\LARGE{$F_3$}};
	\node[vertex](s) at (-3,0) {$s$};
	\node[vertex](w) at (4,0.5) {$w$};
	\node[vertex](y) at (3,0) {$y$};
	\node[clique](K) at (0,1) {$K$};
	
	\draw (s) -- (K) -- (y) -- (w);
	\end{tikzpicture}}
	\caption{\label{fig:F3}Graph $F_3$: $K$ is a clique.}
\end{figure}

\begin{lemma}\label{lemF2}
Consider the graph $F_3$ of Figure \ref{fig:F3} with $|K|=k$ and assume that Bob colored vertex $s$ in the first move of $g^*_B$. Alice wins the game in $F_3$ with $k+1$ colors if and only if she colors vertex $y$ in her first move.
\end{lemma}

\begin{proof}[sketch]
Clearly $s$ is colored 1 and $F_3$ is colored with $k+1$ colors if and only if $y$ and $s$ receive the same color. If Alice colors $y$ in her first move (color 1), she wins with $k+1$ colors. Thus assume Alice does not color $y$. 
Then Bob colors $w$, which receives color 1, forcing different colors for $s$ and $y$.
\end{proof}

\begin{theorem}\label{teo2}
$g^*_B$-1 and $g^*_B$-2 are PSPACE-complete.
\end{theorem}

\begin{proof}[Sketch]
Following similar arguments of Lemma \ref{lem.pspace}, since the number of turns is exactly $n$ and, in each turn, the number of possible moves is at most $n$, we have that both decision problems are PSPACE.
We follow a very similar reduction of Theorem \ref{teo1} (but from POS-CNF-11 instead of POS-DNF-11), including a neighbor $\overline{x_i}$ of degree 1 to each vertex $x_i$ and replacing graph $F_1$ by graph $F_3$ with $|K|=M+3N+24$.
Recall that, in POS-CNF, it is given a CNF formula (conjunctive normal form: conjunction of disjunctions). In POS-CNF-11, there is an additional constraint: the clauses have at most 11 variables.
Figure \ref{fig2b} shows the constructed graph $G$ for the formula $(X_1\vee X_2)\wedge(X_1\vee X_3)\wedge(X_2\vee X_4)\wedge(X_3\vee X_4)$.
In the reduction of this example, we have $N=4$ variables and $M=4$ clauses. The cliques $L_1$ to $L_M$ have $M+3N+19$ vertices each.

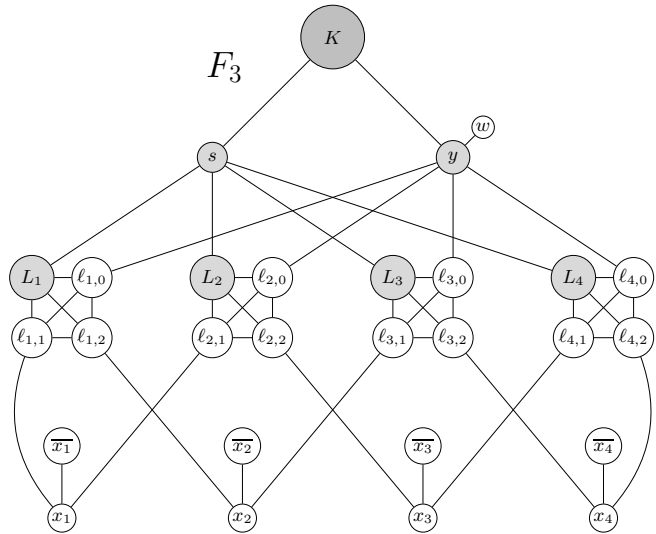
\begin{figure}[h]\label{fig2b}\centering\scalebox{0.8}{
\begin{tikzpicture}[auto]
\tikzstyle{vertex}=[draw,circle,fill=white!25,minimum size=8pt,inner sep=1pt]
\tikzstyle{vertex1}=[draw,circle,fill=white!25,minimum size=12pt,inner sep=2pt]
\tikzstyle{vertex2}=[draw,circle,fill=black!15,minimum size=12pt,inner sep=3pt]
\tikzstyle{vertex3}=[draw,circle,fill=black!25,minimum size=20pt,inner sep=7pt]
\tikzstyle{vertex4}=[draw,circle,fill=black!00,minimum size=20pt,inner sep=7pt]
\node[vertex2] (s) at (2.5,5)   {$s$};
\node[vertex]  (w) at (7,5.5) {$w$};
\node[vertex3] (K) at (4.5,7)   {$K$};
\node[vertex2] (y) at (6.5,5)   {$y$};
\node () at (2.7,6.5) {\LARGE{$F_3$}};

\node[vertex] (x1) at (0,-1) {$x_1$};
\node[vertex] (x2) at (3,-1) {$x_2$};
\node[vertex] (x3) at (6,-1) {$x_3$};
\node[vertex] (x4) at (9,-1) {$x_4$};
\node[vertex1] (x1b) at (0,0.2) {$\overline{x_1}$};
\node[vertex1] (x2b) at (3,0.2) {$\overline{x_2}$};
\node[vertex1] (x3b) at (6,0.2) {$\overline{x_3}$};
\node[vertex1] (x4b) at (9,0.2) {$\overline{x_4}$};

\node[vertex] (L11) at (-0.5,2){$\ell_{1,1}$};
\node[vertex] (L12) at (0.5,2) {$\ell_{1,2}$};
\node[vertex] (L21) at (2.5,2) {$\ell_{2,1}$};
\node[vertex] (L22) at (3.5,2) {$\ell_{2,2}$};
\node[vertex] (L31) at (5.5,2) {$\ell_{3,1}$};
\node[vertex] (L32) at (6.5,2) {$\ell_{3,2}$};
\node[vertex] (L41) at (8.5,2) {$\ell_{4,1}$};
\node[vertex] (L42) at (9.5,2) {$\ell_{4,2}$};

\node[vertex2] (L13) at (-0.5,3){$L_{1}$};
\node[vertex] (L14) at (0.5,3) {$\ell_{1,0}$};
\node[vertex2] (L23) at (2.5,3) {$L_{2}$};
\node[vertex] (L24) at (3.5,3) {$\ell_{2,0}$};
\node[vertex2] (L33) at (5.5,3) {$L_{3}$};
\node[vertex] (L34) at (6.5,3) {$\ell_{3,0}$};
\node[vertex2] (L43) at (8.5,3) {$L_{4}$};
\node[vertex] (L44) at (9.5,3) {$\ell_{4,0}$};

\path[-]
(K) edge (s) edge (y) (y) edge (w)
(x1) edge[bend left] (L11) edge (L21)
(x2) edge (L12) edge (L31)
(x3) edge (L22) edge (L41)
(x4) edge (L32) edge[bend right] (L42)
(L11)edge(L12)edge(L13)edge(L14)(L12)edge(L13)edge(L14)(L13)edge(L14)
(L21)edge(L22)edge(L23)edge(L24)(L22)edge(L23)edge(L24)(L23)edge(L24)
(L31)edge(L32)edge(L33)edge(L34)(L32)edge(L33)edge(L34)(L33)edge(L34)
(L41)edge(L42)edge(L43)edge(L44)(L42)edge(L43)edge(L44)(L43)edge(L44)
(s) edge(L13) edge(L23) edge(L33) edge(L43)
(y) edge(L14) edge(L24) edge(L34) edge(L44)
(x1b) edge (x1) (x2b) edge (x2) (x3b) edge (x3) (x4b) edge (x4);
\end{tikzpicture}}
\caption{Constructed graph $G$ for the formula $(X_1\vee X_2)\wedge(X_1\vee X_3)\wedge(X_2\vee X_4)\wedge(X_3\vee X_4)$. Recall that each vertex $\ell_{j,k}$ represents two true-twins $\ell'_{j,k}$ and $\ell''_{j,k}$, $L_1$, $L_2$, $L_3$, $L_4$ are cliques with $35$ vertices. Bob has a winning strategy avoiding 41 colors in the greedy coloring game.}
\end{figure}

As in Theorem \ref{teo1}, $\chi(G)=M+3N+25$.
With similar arguments in Theorem \ref{teo1}, we obtain the result. The main difference is that, instead coloring a vertex $x_i$ with a color distinct from 1, Alice colors the vertex $\overline{x_i}$ (with color 1).
\end{proof}

As before, we define $\Gamma_g^{A,A}(G)$, $\Gamma_g^{A,B}(G)$, $\Gamma_g^{B,A}(G)$ and $\Gamma_g^{B,B}(G)$: the minimum number of colors in $C$ s.t Alice has a winning strategy in $g^*_{A,A}$, $g^*_{A,B}$, $g^*_{B,A}$ and $g^*_{B,B}$, resp.
With this, we define two decision problem for each game $g^*_{Y,Z}$ with $Y,Z\in\{A,B\}$: given $G$, its chromatic number $\chi(G)$ and an integer $k$,
\begin{itemize}
\item (Problem $g^*_{Y,Z}$-1) $\Gamma_g^{Y,Z}(G)\leq k$\ ? That is, does Alice have a winning strategy in $g_{Y,Z}$ with $k$ colors?
\item (Problem $g^*_{Y,Z}$-2) $\Gamma_g^{Y,Z}(G)=\chi(G)$\ ?
\end{itemize}

\begin{corollary}
For every $Y,Z\in\{A,B\}$, the decision problems $g^*_{Y,Z}$-1 and $g^*_{Y,Z}$-2 are PSPACE-complete.
\end{corollary}

\begin{proof}[Sketch]
Similar to Corollary \ref{corol1}.
\end{proof}


\bibliographystyle{acm}

\begin{thebibliography}{10}

\bibitem{andres-lock19}
{\sc Andres, D., and Lock, E.}
\newblock {Characterising and recognising game-perfect graphs}.
\newblock {\em {Discrete Mathematics \& Theoretical Computer Science} vol. 21
  no. 1\/} (2019).

\bibitem{bodlaender91}
{\sc Bodlaender, H.~L.}
\newblock On the complexity of some coloring games.
\newblock {\em Int J Found Comput Sci 2}, 2 (1991), 133--147.
\newblock {WG-1990}, LNCS {\bf 484} (1990), pp 30-40.

\bibitem{bohman08}
{\sc Bohman, T., Frieze, A., and Sudakov, B.}
\newblock The game chromatic number of random graphs.
\newblock {\em Random Structures \& Algorithms 32}, 2 (2008), 223--235.

\bibitem{zhu19}
{\sc Charpentier, C., Hocquard, H., Sopena, E., and Zhu, X.}
\newblock A connected version of the graph coloring game.
\newblock 9th Slovenian International Conference on Graph Theory (Bled'19)
  arXiv:1907.12276.

\bibitem{rudini19}
{\sc Costa, E., Pessoa, V.~L., Sampaio, R., and Soares, R.}
\newblock {PSPACE}-hardness of two graph coloring games.
\newblock {\em Electronic Notes in Theoretical Computer Science 346\/} (2019),
  333 -- 344.
\newblock proceedings of {LAGOS-2019}.

\bibitem{dinski-zhu99}
{\sc Dinski, T., and Zhu, X.}
\newblock A bound for the game chromatic number of graphs.
\newblock {\em Discrete Mathematics 196}, 1 (1999), 109 -- 115.

\bibitem{faigle93}
{\sc Faigle, U., Kern, U., Kierstead, H., and Trotter, W.}
\newblock On the game chromatic number of some classes of graphs.
\newblock {\em Ars Combinatoria 35\/} (1993), 143 -- 150.

\bibitem{gardner81}
{\sc Gardner, M.}
\newblock Mathematical games.
\newblock {\em Scientific American 244}, 4 (1981), 18 -- 26.

\bibitem{havet13}
{\sc Havet, F., and Zhu, X.}
\newblock The game grundy number of graphs.
\newblock {\em Journal of Combinatorial Optimization 25}, 4 (2013), 752--765.

\bibitem{demaine09}
{\sc Hearn, R.~A., and Demaine, E.~D.}
\newblock {\em Games, Puzzles, and Computation}.
\newblock A. K. Peters, Ltd., Natick, MA, USA, 2009.

\bibitem{kierstad94}
{\sc Kierstead, H.~A., and Trotter, W.~T.}
\newblock Planar graph coloring with an uncooperative partner.
\newblock {\em Journal of Graph Theory 18}, 6 (1994), 569--584.

\bibitem{poset15}
{\sc Krawczyk, T., and Walczak, B.}
\newblock Asymmetric coloring games on incomparability graphs.
\newblock {\em Electronic Notes Discrete Mathematics 49\/} ({Eurocomb} 2015),
  803 -- 811.

\bibitem{nakprasit18}
{\sc Nakprasit, K.~M., and Nakprasit, K.}
\newblock The game coloring number of planar graphs with a specific girth.
\newblock {\em Graphs and Combinatorics 34}, 2 (2018), 349--354.

\bibitem{poscnf78}
{\sc Schaefer, T.~J.}
\newblock On the complexity of some two-person perfect-information games.
\newblock {\em Journal of Computer and System Sciences 16}, 2 (1978), 185 --
  225.

\bibitem{sekiguchi14}
{\sc Sekiguchi, Y.}
\newblock The game coloring number of planar graphs with a given girth.
\newblock {\em Discrete Mathematics 330\/} (2014), 11 -- 16.

\bibitem{game-cactus07}
{\sc Sidorowicz, E.}
\newblock The game chromatic number and the game colouring number of cactuses.
\newblock {\em Information Processing Letters 102}, 4 (2007), 147 -- 151.

\bibitem{zhu99}
{\sc Zhu, X.}
\newblock The game coloring number of planar graphs.
\newblock {\em Journal of Combinatorial Theory, Series B 75}, 2 (1999), 245 --
  258.

\bibitem{zhu00}
{\sc Zhu, X.}
\newblock The game coloring number of pseudo partial k-trees.
\newblock {\em Discrete Mathematics 215}, 1 (2000), 245 -- 262.

\bibitem{zhu08}
{\sc Zhu, X.}
\newblock Refined activation strategy for the marking game.
\newblock {\em Journal of Combinatorial Theory, Series B 98}, 1 (2008), 1 --
  18.

\end{thebibliography}

\end{document}